\documentclass[14pt,reqno,english]{extarticle}
\usepackage{lmodern}
\usepackage{berasans}
\let\origrmdefault\rmdefault
\usepackage[math]{iwona}
\renewcommand{\rmdefault}{\origrmdefault}

\usepackage[T1]{fontenc}
\usepackage[latin9]{inputenc}
\usepackage{color}
\usepackage{babel}
\usepackage{verbatim}
\usepackage{prettyref}
\usepackage{mathrsfs}
\usepackage{amsmath}
\usepackage{amsthm}
\usepackage{amssymb}
\usepackage{setspace}
\usepackage{xargs}[2008/03/08]
\usepackage{microtype}
\usepackage[unicode=true,pdfusetitle,
 bookmarks=true,bookmarksnumbered=false,bookmarksopen=false,
 breaklinks=false,pdfborder={0 0 1},backref=false,colorlinks=true]
 {hyperref}
\hypersetup{
 citecolor=blue,linkcolor=magenta,anchorcolor=green}

\makeatletter

\providecommand{\tabularnewline}{\\}

\numberwithin{equation}{section}
\usepackage{paralist}
\theoremstyle{plain}
\newtheorem{lem}{\protect\lemmaname}
\theoremstyle{plain}
\newtheorem{cor}{\protect\corollaryname}
\theoremstyle{plain}
\newtheorem{thm}{\protect\theoremname}
\theoremstyle{definition}
\newtheorem{defn}{\protect\definitionname}
\theoremstyle{remark}
\newtheorem{rem}{\protect\remarkname}
\theoremstyle{plain}
\newtheorem{prop}{\protect\propositionname}

\@ifundefined{date}{}{\date{}}
\usepackage{upgreek}
\AtBeginDocument{%
  \setlength\abovedisplayskip{3pt plus 1pt minus 1pt}
  \setlength\belowdisplayskip{4pt plus 1pt minus 1pt}}
\newrefformat{cor}{Corollary \ref{#1}}
\newrefformat{prop}{Proposition \ref{#1}}
\newrefformat{conj}{Conjecture \ref{#1}}
\newrefformat{nthm}{Theorem \nameref{#1}}
\newrefformat{def}{Definition \ref{#1}}
\newrefformat{eq}{Eq.(\ref{#1})}
\usepackage{tikz}
\usepackage{subcaption}

\makeatother

\providecommand{\corollaryname}{Corollary}
\providecommand{\definitionname}{Definition}
\providecommand{\lemmaname}{Lemma}
\providecommand{\propositionname}{Proposition}
\providecommand{\remarkname}{Remark}
\providecommand{\theoremname}{Theorem}

\begin{document}
\global\long\def\set#1#2{\left\{  #1\, |\, #2\right\}  }%

\global\long\def\cyc#1{\mathbb{Q}\!\left[\zeta_{#1}\right]}%

\global\long\def\mat#1#2#3#4{\left(\begin{array}{cc}
#1 & #2\\
#3 & #4
\end{array}\right)}%

\global\long\def\Mod#1#2#3{#1\equiv#2\, \left(\mathrm{mod}\, \, #3\right)}%

\global\long\def\inv{^{\,\textrm{-}1}}%

\global\long\def\pd#1{#1^{+}}%

\global\long\def\sym#1{\mathbb{S}_{#1}}%

\global\long\def\fix#1{\mathtt{Fix}\!\left(#1\right)}%

\global\long\def\map#1#2#3{#1\!:\!#2\!\rightarrow\!#3}%

\global\long\def\Map#1#2#3#4#5{\begin{split}#1:#2  &  \rightarrow#3\\
 #4  &  \mapsto#5 
\end{split}
 }%

\global\long\def\fact#1#2{#1\slash#2}%

\global\long\def\Gal#1{\mathtt{Gal}\!\left(#1\right)}%

\global\long\def\fixf#1{\mathbb{Q}\!\left(#1\right)}%

\global\long\def\gl#1#2{\mathsf{GL}_{#2}\!\left(#1\right)}%

\global\long\def\SL{\mathrm{SL}_{2}\!\left(\mathbb{Z}\right)}%

\global\long\def\zn#1{\left(\mathbb{Z}/\!#1\mathbb{Z}\right)^{\times}}%

\global\long\def\sn#1{\mathbb{S}_{#1}}%

\global\long\def\fix#1#2{\mathrm{Fix}_{#1}\left(#2\right)}%

\global\long\def\aut#1{\mathrm{Aut\mathit{\left(#1\right)}}}%

\global\long\def\FA#1{#1^{\prime}}%

\global\long\def\norm#1{\vert#1\vert}%

\global\long\def\ugr#1{#1^{{\scriptscriptstyle \times}}}%

\global\long\def\ipm#1{#1_{{\scriptscriptstyle \pm}}}%

\global\long\def\stab#1#2{#1_{#2}}%

\global\long\def\cent#1#2{\mathtt{C}_{#1}\!\left(#2\right)}%

\global\long\def\om#1{\omega\!\left(#1\right)}%

\global\long\def\sp#1#2{\bigl\langle#1,#2\bigr\rangle}%

\global\long\def\irr#1{\mathtt{Irr}\!\left(#1\right)}%

\global\long\def\sqf#1{#1^{\circ}}%

\global\long\def\duc#1{#1^{{\scriptscriptstyle \perp}}}%

\global\long\def\latg#1{\mathbb{L}\!\left(#1\right)}%

\global\long\def\id{\boldsymbol{\mathfrak{1}}}%

\global\long\def\res#1#2{\mathtt{res}_{#2}^{#1}}%

\global\long\def\rel{\boldsymbol{\nabla}}%
\global\long\def\drel{\boldsymbol{\intercal}}%

\global\long\def\ve{\mathtt{V}}%

\global\long\def\du#1{\boldsymbol{\rel}\!\!\left(#1\right)}%

\global\long\def\star#1{\boldsymbol{\rel}\!\!\left(#1\right)}%

\global\long\def\pcl#1{#1^{\flat}}%

\global\long\def\defl{\partial}%

\global\long\def\cl#1{\boldsymbol{\Updelta}\!\left(#1\right)}%

\global\long\def\ucl#1{\boldsymbol{\varDelta}_{#1}^{\times}}%

\global\long\def\ul#1{\ve_{#1}}%
\global\long\def\mul#1{\mathcal{M}_{#1}}%

\global\long\def\ecl{\mathcal{E}}%

\global\long\def\fcl{\mathcal{\mathscr{L}}}%

\global\long\def\twg{\textrm{twist group}}%

\global\long\def\PCL{\textrm{equilocality class}}%
\global\long\def\PCS{\textrm{neighborhood}}%

\global\long\def\decl{\textrm{deconstruction lattice}}%

\global\long\def\ifc{\mathfrak{I}}%

\global\long\def\FD#1{#1^{\circ}}%

\global\long\def\FE#1{\bigvee#1}%

\global\long\def\FF{}%

\global\long\def\FI{}%

\title{Graphs, lattices and deconstruction hierarchies}
\author{P. Bantay}
\maketitle
\begin{center}
Institute for Theoretical Physics\\
E\"{o}tv\"{o}s Lor{\' a}nd University, Budapest
\par\end{center}
\begin{abstract}
The mathematics underlying the connection between deconstruction lattices
and locality diagrams of conformal models is developed from scratch,
with special emphasis on classification issues. In particular, the
notions of equilocality classes, deflation map, essential vertices
and stem graphs are introduced in order to characterize those graphs
that may arise as locality diagrams. 
\end{abstract}

\section{Introduction}

An important construction procedure in 2D Conformal Field Theory \cite{BPZ,DiFrancesco-Mathieu-Senechal},
leading to new consistent models from known ones, is orbifolding \cite{Dixon_orbifoldCFT,DV3},
which involves identifying states related by some group of symmetries
(the 'twist group'), supplemented by the introduction of new states
(grouped into so-called 'twisted sectors' labeled by the conjugacy
classes of the twist group) to ensure modular invariance. Orbifold
models have found different applications over the years, ranging from
string model building \cite{Dixon_orbifolds1,Dixon_orbifolds2} (based
on the observation that the CFT describing the world-sheet dynamics
of strings moving on a quotient of Minkowski space by some group action
can be obtained in such a way) to the description of second-quantized
strings (via so-called symmetric products orbifolds) \cite{elliptic_genera,Dijkgraaf_disctors,Bantay2003a},
not to mention their role in the celebrated FLM construction of the
Moonshine module \cite{FLM1}, and the (physicist) proof of the congruence
subgroup property of rational conformal models \cite{Bantay2003b}.

Orbifold deconstruction \cite{Bantay2019a}, the inverse procedure
of orbifolding, aims at recognizing whether a given conformal model
could be realized as an orbifold of another one, and if so, to determine
this original model and the relevant $\twg$. That orbifold deconstruction
does not only make sense, but that it can be performed effectively
has been explained in \cite{Bantay2019a,Bantay2020}: starting from
some readily available data (conformal weights, fusion rules and conformal
characters) characterizing a conformal model, one can identify all
possible realizations of the model as an orbifold, possibly up to
some finite (usually rather small) ambiguity that can, as a last resort,
be resolved through a case-by-case analysis. Of course, the computational
need grows steadily with the complexity of the model, but judicious
algorithms allow to deal with models having as much as several hundreds
primary fields.

An interesting aspect of the theory is that there is actually a whole
hierarchy of orbifold deconstructions \cite{Bantay2020a}. This is
natural to some extent: indeed, if a model may be obtained from another
one by orbifolding with respect to some twist group, it may also be
reached in steps, by first orbifolding with respect to a normal subgroup
of the twist group, and then orbifolding the result by the corresponding
factor group. This means that for any deconstruction with a given
$\twg$ one has a full hierarchy of partial deconstructions corresponding
to the different factor groups of the $\twg$, and these form a lattice
isomorphic to the lattice of normal subgroups of the latter. Actually,
the situation is more subtle, because in general one and the same
model can be obtained in genuinely different ways as an orbifold\footnote{That this behavior is the rule rather than the exception is clearly
exemplified by the FLM construction \cite{FLM1} of the Moonshine
module from the Leech lattice model, in which one first constructs
a $\mathbb{Z}_{2}$-orbifold of the latter, and then deconstructs
this holomorphic orbifold: as it turns out, there are two different
non-trivial deconstructions, one giving back (as expected) the Leech
lattice model, while the other results in the Moonshine module.}, hence the different deconstructions of a given model do not form
a lattice, but only some more general kind of ordered structure \cite{Bantay2020a}. 

A major result of \cite{Bantay2020a} is that to each conformal model
is naturally associated an algebraic lattice (with nice properties,
like being modular, which is crucial in view of the group theoretic
interpretation) into which one can embed the whole deconstruction
hierarchy, with the actual deconstructions corresponding to special
lattice elements, the so-called 'twisters' of the model \cite{Bantay2019a}.
Most importantly for us, this '$\decl$' is self-dual, i.e. it comes
equipped with an order-reversing involutory self-map. From a vantage
point, it is fair to say that the $\decl$ is an important combinatorial
characteristic of the conformal model.

Unfortunately, determining the $\decl$ from scratch can be pretty
involved, as one should consider all subsets of primaries, and filter
out those that are closed under the fusion product, a procedure whose
computational cost grows exponentially with the number of primaries,
prohibiting actual computations for models with more than a couple
of dozens of them, while interesting examples involve hundreds, if
not thousands. Fortunately, as explained in \cite{Bantay2021a}, there
is a way out, exploiting the   relationship between self-dual lattices
and undirected graphs (with possible loops): to each graph corresponds
a well-defined self-dual lattice, its 'associated lattice', and there
is  an undirected graph naturally related to any given conformal
model, its so-called 'locality graph' (whose vertices correspond to
the primary fields, with two of them adjacent if the corresponding
primaries are mutually local), whose associated lattice can be shown
to coincide with the $\decl$ $\fcl$. This observation makes it much
more easy to determine the $\decl$, by first considering the locality
graph, and then computing its associated lattice. What is more, by
using the natural isomorphism between the associated lattice of a
graph and that of its deflation, the computations become even simpler. 

The present paper aims at giving a mathematically sound presentation
of the above ideas. We start by describing the basic facts on the
relation between undirected graphs and their associated lattices,
and discuss such notions as equilocality classes and the deflation
map. We then show, by introducing 'duality graphs', that every self-dual
lattice is indeed the associated lattice of some undirected graph.
Next, we turn to the question of describing all (irreducible) graphs
with isomorphic associated lattices: this involves such notions as
essential vertices and the radical of a graph. Finally, we present
the physics motivation behind all the foregoing, by showing that the
associated lattice of the locality graph does indeed coincide with
the $\decl$ of the conformal model, and expound on some of the most
interesting consequences of this result.

\global\long\def\set#1#2{\left\{  #1\, |\, #2\right\}  }%

\global\long\def\cyc#1{\mathbb{Q}\!\left[\zeta_{#1}\right]}%

\global\long\def\mat#1#2#3#4{\left(\begin{array}{cc}
#1 & #2\\
#3 & #4
\end{array}\right)}%

\global\long\def\Mod#1#2#3{#1\equiv#2\, \left(\mathrm{mod}\, \, #3\right)}%

\global\long\def\inv{^{\,\textrm{-}1}}%

\global\long\def\pd#1{#1^{+}}%

\global\long\def\sym#1{\mathbb{S}_{#1}}%

\global\long\def\fix#1{\mathtt{Fix}\!\left(#1\right)}%

\global\long\def\map#1#2#3{#1\!:\!#2\!\rightarrow\!#3}%

\global\long\def\Map#1#2#3#4#5{\begin{split}#1:#2  &  \rightarrow#3\\
 #4  &  \mapsto#5 
\end{split}
 }%

\global\long\def\fact#1#2{#1\slash#2}%

\global\long\def\Gal#1{\mathtt{Gal}\!\left(#1\right)}%

\global\long\def\fixf#1{\mathbb{Q}\!\left(#1\right)}%

\global\long\def\gl#1#2{\mathsf{GL}_{#2}\!\left(#1\right)}%

\global\long\def\SL{\mathrm{SL}_{2}\!\left(\mathbb{Z}\right)}%

\global\long\def\zn#1{\left(\mathbb{Z}/\!#1\mathbb{Z}\right)^{\times}}%

\global\long\def\sn#1{\mathbb{S}_{#1}}%

\global\long\def\fix#1#2{\mathrm{Fix}_{#1}\left(#2\right)}%

\global\long\def\aut#1{\mathrm{Aut\mathit{\left(#1\right)}}}%

\global\long\def\FA#1{#1^{\prime}}%

\global\long\def\norm#1{\vert#1\vert}%

\global\long\def\ugr#1{#1^{{\scriptscriptstyle \times}}}%

\global\long\def\ipm#1{#1_{{\scriptscriptstyle \pm}}}%

\global\long\def\stab#1#2{#1_{#2}}%

\global\long\def\cent#1#2{\mathtt{C}_{#1}\!\left(#2\right)}%

\global\long\def\om#1{\omega\!\left(#1\right)}%

\global\long\def\sp#1#2{\bigl\langle#1,#2\bigr\rangle}%

\global\long\def\irr#1{\mathtt{Irr}\!\left(#1\right)}%

\global\long\def\sqf#1{#1^{\circ}}%

\global\long\def\duc#1{#1^{{\scriptscriptstyle \perp}}}%

\global\long\def\latg#1{\mathbb{L}\bigl(\!#1\!\bigr)}%

\global\long\def\id{\boldsymbol{\mathfrak{1}}}%

\global\long\def\res#1#2{\mathtt{res}_{#2}^{#1}}%

\global\long\def\rel{\boldsymbol{\nabla}}%
\global\long\def\drel{\boldsymbol{\intercal}}%

\global\long\def\ve{\mathtt{V}}%

\global\long\def\du#1{\boldsymbol{\rel}\!\!\left(#1\right)}%

\global\long\def\star#1{\boldsymbol{\rel}\!\!\left(#1\right)}%

\global\long\def\pcl#1{#1^{\flat}}%

\global\long\def\defl{\partial}%

\global\long\def\cl#1{\boldsymbol{\Updelta}\!\left(#1\right)}%

\global\long\def\ucl#1{\boldsymbol{\varDelta}_{#1}^{\times}}%

\global\long\def\ul#1{\ve_{#1}}%
\global\long\def\mul#1{\mathcal{M}_{#1}}%

\global\long\def\ecl{\mathscr{E}}%

\global\long\def\PCL{\textrm{equilocality class}}%
\global\long\def\PCS{\textrm{neighborhood}}%

\global\long\def\FD#1{#1^{\circ}}%

\global\long\def\FE#1{\bigvee#1}%

\global\long\def\FF{}%

\global\long\def\FI{}%

\section{Graphs and their lattices\label{sec:Graphs-and}}

Let $\rel$ denote an undirected graph \cite{Balakrishnan1997,Bollobas2002}
with (finite) vertex set $\ve$. To each vertex $x\!\in\!\ve$ is
associated its neighborhood $\star x$, the set of all vertices adjacent
to it. More generally, for any subset $X\!\subseteq\!\ve$ its dual
\begin{equation}
\du X=\set{x\!\in\!\ve}{X\!\subseteq\!\star x}={\displaystyle \bigcap_{x\in X}}\star x\label{eq:dualdef}
\end{equation}
is the set of all vertices that are adjacent to each element of $X$.
\begin{lem}
\label{lem:dualsym}For any $X,Y\!\subseteq\!\ve$ one has $X\!\subseteq\!\du Y$
iff $Y\!\subseteq\!\du X$.
\end{lem}
\begin{proof}
$X\!\subseteq\!\du Y$ means that all vertices in $X$ are adjacent
to all vertices in $Y$, i.e. $Y\!\subseteq\!\du X$ by the symmetry
of the adjacency relation.
\end{proof}
\begin{lem}
\begin{singlespace}
\label{lem:dualmeet}$\du{X\!\cup\!Y}\!=\!\du X\cap\du Y$ for $X,Y\!\subseteq\!\ve$.
\end{singlespace}
\end{lem}
\begin{proof}
$x\!\in\!\du{X\!\cup\!Y}$ iff $X\!\cup\!Y\!\subseteq\!\star x$,
i.e. $X\!\subseteq\!\star x$ and $Y\!\subseteq\!\star x$, and this
is equivalent to $x\!\in\!\du X$ and $x\!\in\!\du Y$, proving the
claim.
\end{proof}
\begin{lem}
\label{lem:dualorder} The assignment $X\!\mapsto\!\du X$ is order-reversing,
i.e. $X\!\subseteq\!Y$ implies $\du Y\!\subseteq\!\du X$.
\end{lem}
\begin{proof}
One has $X\!\subseteq\!Y$ iff $X\!\cup\!Y\!=\!Y$, and by \prettyref{lem:dualmeet}
this implies $\du Y\!=\!\du{X\!\cup\!Y}\!=\!\du X\cap\du Y$, that
is $\du Y\!\subseteq\!\du X$.
\end{proof}
\begin{lem}
\label{lem:closure}For a subset $X\!\subseteq\!\ve$ let
\begin{equation}
\cl X\!=\!\du{\du X}\!=\!\set{x\!\in\!\ve}{\du X\!\subseteq\!\star x}\label{eq:cldef}
\end{equation}
Then
\begin{align}
X\! & \subseteq\!\cl X\label{eq:clext}\\
\cl{\du X} & \!=\!\du{\cl X}\!=\!\du X\label{eq:cldu}\\
\cl X\! & \subseteq\!\cl Y\label{eq:clord}
\end{align}
 for $X\!\subseteq\!Y\!\subseteq\!\ve$.
\end{lem}
\begin{proof}
\prettyref{eq:clext} follows from \prettyref{lem:dualsym} with $Y\!=\!\du X$.
As to \prettyref{eq:cldu}, substituting $X$ by $\du X$ in \prettyref{eq:clext}
gives $\du X\!\subseteq\!\cl{\du X}$, while applying $\rel$ to both
sides of it and using \prettyref{lem:dualorder} leads to $\du{\cl X}\!\subseteq\!\du X$.
Since the composition of mappings is associative $\du{\cl X}\!=\!\cl{\du X}$,
hence $\du X\!\subseteq\!\cl{\du X}\!=\!\du{\cl X}\!\subseteq\!\du X$,
proving \prettyref{eq:cldu}. Finally, \prettyref{eq:clord} is a
consequence of \prettyref{lem:dualorder} applied twice.
\end{proof}
\begin{cor}
\label{cor:closure}$\cl{\cl X}\!=\!\cl X$ for any $X\!\subseteq\!\ve$,
hence the assignment $X\!\mapsto\!\cl X$ is a closure operator on
subsets of $\ve$.
\end{cor}
\begin{proof}
Substituting $X$ by $\du X$ in \prettyref{eq:cldu} gives $\cl{\cl X}\!=\!\cl X$,
and combined with Eqs.\eqref{eq:clext} and \eqref{eq:clord} this
proves that $X\!\mapsto\!\cl X$ is indeed a closure operator.
\end{proof}
Recall \cite{Gratzer2011} that a lattice $L$ is a partially ordered
set in which every collection $X\!\subseteq\!L$ of lattice elements
has a least upper bound (their join ${\displaystyle {\textstyle \bigvee}}X$)
and a greatest lower bound (their meet ${\displaystyle {\textstyle \bigwedge}}X$).
The lattice is bounded if it has a maximal and a minimal element,
and it is self-dual if there exists a map $a\mapsto\duc a$ of $L$
onto itself (the 'duality map'), which is involutive, i.e. $\duc{(\duc a)}\!=\!a$,
and order-reversing, i.e. $a\!\leq\!b$ implies $\duc b\!\leq\!\duc a$
for all $a,b\!\in\!L$. Note that in a self-dual lattice join and
meet are related by the duality map via de Morgan's law $\duc{\left(a\!\vee\!b\right)}\!=\!\duc a\!\wedge\!\duc b$%
. Two self-dual lattices $L_{1}$ and $L_{2}$ are isomorphic iff
there exists a lattice isomorphism (an order-preserving bijective
map whose inverse is also order-preserving) $\map{\phi}{L_{1}}{L_{2}}$
compatible with the respective duality maps, i.e. such that $\phi\!\left(\duc a\right)\!=\!\duc{\phi\!\left(a\right)}$
for all $a\!\in\!L_{1}$. In particular, the automorphism group $\aut L$
of the self-dual lattice $L$ consists of those lattice automorphisms
of $L$ that commute with the duality map.
\begin{thm}
\label{thm:graphlat}The collection $\latg{\rel}=\!\set{X\!\subseteq\!\ve}{\cl X\!=\!X}$
ordered by inclusion is a self-dual lattice, with duality map given
by $X\!\mapsto\!\du X$.
\end{thm}
\begin{proof}
That $\latg{\rel}$ is a lattice is a general feature of closure operators,
so what we have to prove is that the assignment $X\!\mapsto\!\du X$
is a duality map. That $\du X\!\in\!\latg{\rel}$ if $X\!\in\!\latg{\rel}$
follows from \prettyref{eq:cldu}, hence $X\!\mapsto\!\du X$ does
indeed map $\latg{\rel}$ into itself; it is order-reversing according
to \prettyref{lem:dualorder}, and involutive by \prettyref{eq:cldef}.
\end{proof}
The lattice $\latg{\rel}$ has as maximal element the set of all vertices,
while its minimal element is the set of those (so-called 'universal')
vertices that are adjacent to every vertex. The meet of $X,Y\!\in\!\latg{\rel}$
is their intersection $X\!\cap\!Y$, while their join can be expressed
using de Morgan's law as
\begin{equation}
X\vee Y=\du{\du X\cap\du Y}\label{eq:join}
\end{equation}

\begin{lem}
\label{lem:dualclosed}$\cl X\!=\!X$ iff $X\!=\!\du Y$ for some
$Y\!\subseteq\!\ve$, hence every closed set $X\!\in\!\latg{\rel}$
is an intersection of vertex-neighborhoods.
\end{lem}
\begin{proof}
If $\cl X\!=\!X$, then $X\!=\!\du Y$ holds with $Y\!=\!\du X$ according
to \prettyref{eq:cldef}. Conversely, $X\!=\!\du Y$ implies $\cl X\!=\!\cl{\du Y}\!=\!\du Y\!=\!X$
by \prettyref{eq:cldu}, proving the claim.
\end{proof}
\begin{lem}
\label{lem:dpc}The least element of $\latg{\rel}$ that contains
$x\!\in\!\ve$ is
\begin{equation}
\cl x\!=\!\!\set{y\!\in\!\ve}{\star x\!\subseteq\!\star y}\label{eq:dpc}
\end{equation}
Consequently,  for all $X\!\subseteq\!\ve$
\begin{equation}
\cl X=\bigvee_{x\in X}\cl x\label{eq:cljoin}
\end{equation}
\end{lem}
\begin{proof}
It is clear that $x\!\in\!\cl x$, and because $\cl x$ equals $\du{\star x}$
according to \prettyref{eq:dualdef}, it belongs to $\latg{\rel}$
by \prettyref{lem:dualclosed}. On the other hand, if $X\!\in\!\latg{\rel}$
contains $x\!\in\!\ve$, then $\du X\!\subseteq\!\star x$ by \prettyref{eq:dualdef},
hence $\cl x\!\subseteq\!\cl X\!=\!X$ by \prettyref{lem:dualorder},
proving that $\cl x$ is indeed the least element of $\latg{\rel}$
that contains $x\!\in\!\ve$. As to \prettyref{eq:cljoin}, one has
\[
\cl X\!=\!\du{\du X}\!=\!\rel\Bigl(\bigcap_{x\in X}\!\star x\Bigr)\!=\!\bigvee_{x\in X}\du{\star x}\!=\!\bigvee_{x\in X}\cl x
\]
taking into account \prettyref{eq:join}.
\end{proof}

\begin{lem}
\label{lem:graphiso}An isomorphism between two graphs induces an
isomorphism between their associated lattices.
\end{lem}
\begin{proof}
The graphs $\rel_{1}$ and $\rel_{2}$ are isomorphic in case there
exists a one-to-one correspondence $\phi$ between their vertices
such that two vertices of $\rel_{1}$ are adjacent iff their images
are adjacent vertices of $\rel_{2}$, hence 
\begin{multline}
\phi\!\left(\rel_{1}\!\left(X\right)\right)\!=\!\set{\phi\!\left(y\right)}{y\!\in\!\rel_{1}\!\left(x\right)\textrm{ for all }x\!\in\!X}\\
\!=\!\set y{y\!\in\!\rel_{2}\!\left(\phi\!\left(x\right)\right)\textrm{ for all }x\!\in\!X}\!=\!\rel_{2}\!\left(\phi\!\left(X\right)\right)\label{eq:griso}
\end{multline}
for any set $X$ of vertices of $\rel_{1}$, implying that $\phi\!\left(X\right)$
belongs to $\latg{\rel_{2}}$ precisely when $X\!\in\!\latg{\rel_{1}}$
by \prettyref{lem:dualclosed}. As a result, $X\!\mapsto\!\phi\!\left(X\right)$
is a one-to-one order-preserving map from $\latg{\rel_{1}}$ to $\latg{\rel_{2}}$,
and it is compatible with the respective duality maps thanks to \prettyref{eq:griso}.
\end{proof}
As the above result shows, the isomorphism class of the self-dual
lattice $\latg{\rel}$ is completely determined by the isomorphism
class of the graph $\rel$, but the converse is far from being true:
there are infinitely many non-isomorphic graphs with isomorphic associated
lattices. To understand what happens, we need the following notion.
\begin{defn}
\label{def:equilocal}Two vertices of a graph are equilocal if their
vertex-neighborhoods coincide. 
\end{defn}
That is, the vertices $x,y\!\in\!\ve$ of the graph $\rel$ are equilocal
when $\star x\!=\!\star y$. Equilocality is clearly an equivalence
relation, whose equivalence classes (the '$\PCL$es') partition the
set of all vertices. Equilocality is compatible with adjacency in
the sense that if two vertices are adjacent, then every pair of vertices
from their respective $\PCL$es are also adjacent; put another way,
the $\PCL$es provide a modular partition \cite{Balakrishnan1997}
of the graph.
\begin{lem}
\label{lem:pcl}The neighborhood of a vertex is a union of $\PCL$es.
\end{lem}
\begin{proof}
We have to show that if $y\!\in\!\star x$ and $\du y\!=\!\du z$,
then $z\!\in\!\star x$. But this is immediate, since (by the symmetric
nature of the adjacency relation) $y\!\in\!\star x$ iff $x\!\in\!\star y$
iff $x\!\in\!\star z$ iff $z\!\in\!\star x$.
\end{proof}
\begin{cor}
\label{cor:ecldecomp}Any $X\!\in\!\latg{\rel}$ is a union of $\PCL$es.
\end{cor}
\begin{proof}
This follows at once from Lemmas \ref{lem:dualclosed} and \ref{lem:pcl}.
\end{proof}
\begin{defn}
\label{def:defl}The \emph{deflation} $\pcl{\rel}$ of an undirected
graph $\rel$ is the graph whose vertices are the $\PCL$es of $\rel$,
with two of them adjacent if their elements are adjacent in $\rel$.
\end{defn}
\begin{lem}
\label{lem:deflirr}The deflation $\pcl{\rel}$ is an irreducible
graph, i.e. each of its $\PCL$es contains precisely one vertex.
\end{lem}
\begin{proof}
The neighborhood in $\pcl{\rel}$ of an $\PCL$ $\ecl$ consists of
those $\PCL$es all of whose elements are adjacent to every element
of $\ecl$, hence the union of these classes gives the neighborhood
of any element of $\ecl$. This means that, should two $\PCL$es have
the same neighborhood in $\pcl{\rel}$, their respective elements
would also have equal neighborhoods in $\rel$, hence the two classes
would coincide.
\end{proof}
Clearly, the graph determines both its deflation and the collection
of its $\PCL$es, and conversely, the knowledge of the deflation and
of the $\PCL$es determines completely the graph. In this sense, the
study of arbitrary undirected graphs may be reduced to that of irreducible
ones that are isomorphic to their deflation.

For $X\!\in\!\latg{\rel}$, let's consider the set $\defl\!\left(X\right)$
of those $\PCL$es that are contained in $X$: this is well-defined
according to \prettyref{cor:ecldecomp}, and satisfies trivially ${\displaystyle {\textstyle \bigcup}}\!~\defl\!\left(X\right)\!=\!X$.
\begin{lem}
\label{lem:defldu}For any $X\!\in\!\latg{\rel}$ one has
\begin{equation}
\pcl{\rel}\!\left(\defl\!\left(X\right)\right)\!=\!\defl\!\left(\du X\right)\label{eq:defldu}
\end{equation}
\end{lem}
\begin{proof}
A class belongs to $\pcl{\rel}\!\left(\defl\!\left(X\right)\right)$
precisely if it is adjacent with every $\PCL$ contained in $\defl\!\left(X\right)$,
that is, if all of its elements are adjacent to every element of $X$,
i.e. if it is contained in $\du X$.
\end{proof}
\begin{cor}
\textup{$\defl\!\left(X\right)\!\in\!\latg{\,\pcl{\rel}\,}$ for every
}$X\!\in\!\latg{\rel}$.
\end{cor}
\begin{proof}
This follows from \prettyref{lem:dualclosed}, since $\defl\!\left(X\right)\!=\!\pcl{\rel}\!\left(\defl\!\left(\du X\right)\right)$
for any $X\!\in\!\latg{\rel}$ according to \prettyref{lem:defldu}.
\end{proof}
\begin{thm}
\label{thm:defliso}The map\textup{ $X\!\mapsto\!\defl\!\left(X\right)$}
that assigns to each $X\!\in\!\latg{\rel}$ the collection of all
$\PCL$es contained in it is an isomorphism between the self-dual
lattices $\latg{\rel}$ and $\latg{\pcl{\,\rel}\,}$.
\end{thm}
\begin{proof}
The map $X\!\mapsto\!\defl\!\left(X\right)$ between $\latg{\rel}$
and $\latg{\pcl{\rel}}$ is obviously order-preserving and injective.
It is also surjective (hence an isomorphism) because $X\!=\!\defl\!\left(\cup X\right)$
for every $X\!\in\!\latg{\pcl{\rel}}$, and it is compatible with
the respective duality maps according to \prettyref{eq:defldu}. 
\end{proof}
The importance of \prettyref{thm:defliso} is that it allows to reduce
the study of the lattice $\latg{\rel}$ to that of $\latg{\pcl{\rel}}$
and the deflation isomorphism $X\!\mapsto\!\defl\!\left(X\right)$.
In particular, the structure of $\latg{\rel}$ as a self-dual lattice
is completely captured by that of $\latg{\pcl{\rel}}$. For this reason,
we shall usually restrict our attention to irreducible graphs (unless
indicated otherwise). It is also of great practical importance for
explicit computations, since the cost of determining the equilocality
classes and the deflation map grows polynomially with size (= number
of vertices), while in case of the associated lattice the growth is
exponential.

Let's go back to the question of classifying all non-isomorphic graphs
with isomorphic associated lattices. \prettyref{thm:defliso} goes
a long way in answering it, since it shows that graphs with isomorphic
deflations have isomorphic associated lattices. Since the deflation
of a graph is always irreducible according to \prettyref{lem:deflirr},
the problem can be reduced to that of classifying all non-isomorphic
irreducible graphs with isomorphic associated lattices. To attack
this problem, we need to show first that every self-dual lattice is
the associated lattice of some irreducible graph.

\section{The duality graph}

As we have seen in the previous section, one can associate to any
undirected graph a self-dual lattice. It is now time to look at the
inverse procedure that associates to a self-dual lattice an undirected
graph whose associated lattice is isomorphic with the one that we
started with.

\global\long\def\ve{L}%
\global\long\def\drel{\boldsymbol{\nabla}_{L}}%
\global\long\def\phia#1{L_{#1}}%
\global\long\def\dugr{\textrm{duality graph}}%
\global\long\def\dgr{\mathtt{\boldsymbol{\nabla}}{}_{L}}%

\begin{defn}
\label{def:dugr}Let $L$ be a self-dual lattice with duality map
$a\!\mapsto\!\duc a$. The $\dugr$ $\dgr$ is the undirected graph
whose vertices are the elements of $L$, with $a,b\!\in\!L$ adjacent
if $a\!\leq\!\duc b$.
\end{defn}
\begin{rem}
One could consider a variant $\FA{\dgr}$ of the $\dugr$ that is
obtained by reversing the order relation in the definition of $\dgr$,
i.e. by declaring the elements $a,b\!\in\!L$ adjacent if $a\!\geq\!\duc b$,
but this won't lead to anything new, as the duality map $a\!\mapsto\!\duc a$
provides a natural graph isomorphism between $\dgr$ and $\FA{\dgr}$.
\end{rem}
\begin{lem}
\label{lem:dungjoin}The image of $a\!\in\!L$ under the duality map
equals the join of its vertex-neighborhood $\dgr\!\left(a\right)$,
i.e. $\duc a\!=\!\FE{\dgr}\!\left(a\right)$.
\end{lem}
\begin{proof}
This follows at once from $\dgr\!\left(a\right)\!=\!\set{b\!\in\!L}{b\!\leq\!\duc a}$.
\end{proof}
\begin{cor}
\label{cor:dugrirr}The $\dugr$ $\dgr$ is irreducible.
\end{cor}
\begin{proof}
By \prettyref{lem:dungjoin}, $\dgr\!\left(a\right)\!=\!\dgr\!\left(b\right)$
implies $\duc a\!=\!\duc b$, i.e. $a\!=\!b$. 
\end{proof}
\begin{lem}
\label{lem:dugrdu}$\dgr\!\left(X\right)\!=\!\dgr\!\left(\FE X\right)$
for $X\!\subseteq\!L$.
\end{lem}
\begin{proof}
Indeed,
\[
\dgr\!\left(\FE X\right)\!=\!\set{b\!\in\!L}{\FE X\!\leq\!\duc b}\!=\!\set{b\!\in\!L}{a\!\leq\!\duc b\textrm{ for all }a\!\in\!X}\!=\!\dgr\!\left(X\right)
\]
by the very definition of the join $\FE X$.
\end{proof}
\begin{cor}
\label{cor:dugrels}Every element $X\!\in\!\latg{\dgr}$ is a vertex-neighborhood.
\end{cor}
\begin{proof}
By \prettyref{lem:dualclosed}, any $X\!\in\!\latg{\dgr}$ can be
written as $X\!=\!\dgr\!\left(Y\right)$ with a suitable $Y\!\subseteq\!L$,
hence $X\!=\!\dgr\!\left(a\right)$ by \prettyref{lem:dugrdu} (with
$a\!=\!\FE Y$).
\end{proof}
\begin{thm}
\label{thm:dugr}Any self-dual lattice is isomorphic to the associated
lattice of its $\dugr$.
\end{thm}
\begin{proof}
Let $L$ denote a self-dual lattice. We claim that the map
\[
\Map{\Phi}L{\latg{\dgr}}a{\dgr\!\left(\duc a\right)}
\]
that assigns to each lattice element the vertex-neighborhood of its
dual provides an isomorphism between $L$ and the associated lattice
$\latg{\dgr}$ of its $\dugr$ $\dgr$. Clearly, $\Phi$ is an order-preserving
map, being the composite of two (order-reversing) duality maps, and
it is surjective by \prettyref{cor:dugrels}. Because the map
\[
\Map{\Psi}{\latg{\dgr}}LX{\FE X}
\]
is inverse to $\Phi$ according to \prettyref{lem:dungjoin}, $\Phi$
is actually bijective, hence a lattice isomorphism. Finally, $\Phi$
is compatible with the respective duality maps since clearly $\Phi\!\left(\duc a\right)\!=\!\dgr\!\left(a\right)$
for $a\!\in\!L$.
\end{proof}
It follows from the above that every self-dual lattice is the associated
lattice of some irreducible graph (e.g. of its own $\dugr$). But
it should be emphasized that this correspondence is far from being
one-to-one, for there could (and usually does) exist several non-isomorphic
irreducible graphs with isomorphic associated lattices. To get control
over them, we shall need the notions of essential vertices and the
radical of a graph, to be introduced in the next section.

\section{Essential vertices and the radical}

\global\long\def\sw{\mathtt{W}}%
\global\long\def\igr#1#2{#1^{#2}}%
\global\long\def\sigr{\igr{\rel}{\sw}}%
\global\long\def\res{\mathtt{res}_{\sw}}%
\global\long\def\ve{\mathtt{V}}%
\global\long\def\sgr#1{\igr{\rel}{\sw_{#1}}}%
\global\long\def\isl{\mathbb{L}_{\sw}\left(\rel\right)}%

As we have seen previously, one can associate a self-dual lattice
to any undirected graph, and it is natural to ask to what extent does
this associated lattice characterize the graph itself. To some extent
this has been answered by \prettyref{thm:defliso}, showing that graphs
with isomorphic deflations have isomorphic associated lattices, allowing
to reduce the question to the case of irreducible graphs. But there
are usually several different irreducible graphs with isomorphic associated
lattices: in particular, \prettyref{thm:dugr} asserts that this is
the case for any irreducible graph that is not isomorphic with the
$\dugr$ of its associated lattice. To settle this issue, we have
to look at the associated lattices of induced subgraphs.

Recall that the induced subgraph $\sigr$ corresponding to a collection
$\sw$ of vertices of the (undirected) graph $\rel$ has as vertices
the elements of $\sw$, with two of them adjacent precisely when they
are adjacent as vertices of $\rel$: this means that $\sigr\!\left(x\right)\!=\!\rel\!\left(x\right)\!\cap\!\sw$
for $x\!\in\!\sw$, and in general $\sigr\!\left(X\right)\!=\!\rel\!\left(X\right)\!\cap\!\sw$
for any $X\!\subseteq\!\sw$. Note that an induced subgraph of an
irreducible graph is not necessarily irreducible. The importance of
induced subgraphs stems from the following result.
\begin{lem}
\label{lem:dugrsub}Every irreducible graph is isomorphic with an
induced subgraph of the $\dugr$ of its associated lattice.
\end{lem}
\begin{proof}
Let $L\!=\!\latg{\rel}$ denote the associated lattice of the irreducible
graph $\rel$, and consider the collection $\sw\!=\!\set{\cl x\!}{x\!\in\!\ve}\!\subseteq\!L$
of lattice elements. The vertices of $\rel$ are in one-to-one correspondence
with the elements of $\sw$ since $\cl x\!=\!\cl y$ implies $x\!=\!y$
by irreducibility, and the vertices $\cl x$ and $\cl y$ of $\dgr^{\sw}$
are adjacent iff $\cl x\!\subseteq\!\du{\cl y}\!=\!\du y$, which
is equivalent to $x\!\in\!\du y$ by \prettyref{lem:dpc}, i.e. the
adjacency of the vertices $x$ and $y$ of $\rel$, proving that $\dgr^{\sw}$
is indeed isomorphic to $\rel$.
\end{proof}
The following result will be essential in the proof of \prettyref{thm:isgisom}.
\begin{lem}
\label{lem:isgr}If $\sw$ denotes a collection of vertices of the
 graph $\rel$, then $\isl\!=\!\set{\rel\!\left(X\right)\!}{X\!\subseteq\!\sw}$
is a subset of the associated lattice $\latg{\rel}$, and $\latg{\sigr\,}\!=\!\set{X\!\cap\!\sw}{X\!\in\!\isl}$.
Moreover, $\cl{X\!\cap\!\sw}\!=\!X$ and $\sigr\!\left(X\!\cap\!\sw\right)\!=\!\sigr\!\left(X\right)$
in case $X\!\in\!\latg{\rel}$ satisfies $\du X\!\in\!\isl$.
\end{lem}
\begin{proof}
$X\!\in\!\isl$ means that there exists $X^{*}\!\subseteq\!\sw$ such
that $X\!=\!\du{X^{*}}$, hence $X\!\in\!\latg{\rel}$ and $X\!\cap\!\sw\!=\!\sigr\!\left(X^{*}\right)\!\in\!\latg{\sigr\,}$
by \prettyref{lem:dualclosed}. As to the second statement, $\du X\!\in\!\isl$
if there exists $X^{*}\!\subseteq\!\sw$ such that $\du X\!=\!\du{X^{*}}$,
which implies $X^{*}\!\subseteq\!\cl{X^{*}}\!=\!X$ according to \prettyref{eq:clext},
i.e. $X^{*}\!\subseteq\!X\!\cap\!\sw\!\subseteq\!X$. As a consequence,
$X\!=\!\cl{X^{*}}\!\subseteq\!\cl{X\!\cap\!\sw}\!\subseteq\!X$ by
\prettyref{eq:cldu}, proving that indeed $\cl{X\!\cap\!\sw}\!=\!X$,
and applying $\rel$ to both sides leads to $\du{X\!\cap\!\sw}\!=\!\du X$,
hence $\sigr\!\left(X\!\cap\!\sw\right)\!=\!\sigr\!\left(X\right)$.
\end{proof}
Let's recall that an element of a lattice is called (completely) meet-irreducible
(resp. join-irreducible), if any collection of lattice elements whose
meet (resp. join) is the given element necessarily contains that element;
note that in a self-dual lattice the duality map interchanges meet-irreducible
elements with join-irreducible ones. A meet-(resp. join-)irreducible
decomposition of a lattice element $a$ is a collection $A$ of meet-(resp.
join-)irreducible elements whose meet (resp. join) equals $a$, and
the decomposition is irredundant if no proper subset of $A$ has this
last property. In a finite lattice all lattice elements have (not
necessarily unique) irredundant irreducible decompositions.
\begin{lem}
\label{lem:pclunion}Any meet-irreducible element of the lattice $\latg{\rel}$
is the $\PCS$ $\star x$ of some vertex $x\!\in\!\ve$.
\end{lem}
\begin{proof}
According to \prettyref{eq:dualdef}, there should exist some $x\!\in\!\du X$
such that $X\!=\!\star x$ in case $X\!\in\!\latg{\rel}$ is meet-irreducible.
\end{proof}
\global\long\def\ess#1{\mathcal{E}\!\left(#1\right)}%
\global\long\def\rad#1{\sqrt{#1}}%
\global\long\def\stg{\textrm{stem graph}}%

The above result justifies the following notion.
\begin{defn}
\label{def:essdef}A vertex $x$ of an undirected graph $\rel$ is
essential if its neighborhood $\rel\!\left(x\right)$ is a meet-irreducible
(or, what is the same, $\cl x$ is a join-irreducible) element of
the associated lattice $\latg{\rel}$. A \emph{$\stg$} is an irreducible
graph all of whose vertices are essential.

We shall denote by $\ess{\rel}$ the collection of all essential vertices
of the graph $\rel$. Notice that either none or all vertices in an
equilocality class are essential, hence it makes sense to speak of
essential classes. What is more, the deflation map provides a one-to-one
correspondence between the essential classes and the essential vertices
of the deflation, i.e. $\ess{\pcl{\rel}}\!=\!\defl\ess{\rel}$.

For duality graphs of self-dual lattices one has the following lattice-theoretic
characterization of essential vertices.
\end{defn}
\begin{lem}
\label{lem:dugress}A vertex of the $\dugr$ of a self-dual lattice
is essential iff it is a join-irreducible lattice element.
\end{lem}
\begin{proof}
According to \prettyref{def:essdef}, $a\!\in\!L$ belongs to $\ess{\dgr}$
iff $\dgr\!\left(a\right)$ is a meet-irreducible element of $\latg{\dgr}$,
i.e. $\dgr\!\left(X\right)\!=\!\dgr\!\left(a\right)$ with $X\!\subseteq\!L$
would imply $a\!\in\!X$; since $\dgr\!\left(X\right)\!=\!\dgr\!\left(\FE X\right)$
by \prettyref{lem:dugrdu}, and because the $\dugr$ $\dgr$ is irreducible
by \prettyref{cor:dugrirr}, this is equivalent to the requirement
that $a\!=\!\FE X$ should imply $a\!\in\!X$, i.e. that $a\!\in\!L$
is join-irreducible, as claimed.
\end{proof}
While our definition of essential vertices relies on lattice-theoretic
notions, there is a purely graph-theoretic characterization of them.
\begin{prop}
\label{prop:irrcrit}A vertex is essential iff its neighborhood is
properly contained in the intersection of all vertex-$\PCS$s that
properly contain it.
\end{prop}
\begin{proof}
Since $y\!\in\!\cl x$ iff $\du x\!\subseteq\!\du y$ for any vertex
$x\!\in\!\ve$, the set $\ucl x\!=\!\set{y\!\in\!\cl x\!\!}{\star y\!\neq\!\star x\!}$
consists of all those elements $y\!\in\!\cl x$ that do not belong
to the equilocality class of $x$, hence it is a proper subset of
$\cl x$, and its dual $\du{\ucl x}$ equals the intersection of all
the vertex $\PCS$s that properly contain the neighborhood $\du x$
of $x$, hence $\du x\!\subseteq\!\du{\ucl x}$, and this containment
is proper iff the containment $\cl{\ucl x}\!\subseteq\!\cl x$ is
proper. Since the difference $\cl x\!\setminus\!\ucl x\!=\!\set y{\star y\!=\!\star x}$
is nothing but the $\PCL$ of $x$, and because $\cl{\ucl x}$ is
a union of equilocality classes by \prettyref{lem:pcl}, either $\cl{\ucl x}\!=\!\ucl x$
or $\cl{\ucl x}\!=\!\cl x$, hence the containment $\cl{\ucl x}\!\subseteq\!\cl x$
is proper iff $\ucl x\!\in\!\latg{\rel}$, and we claim that $\ucl x\!\in\!\latg{\rel}$
precisely when $\cl x$ is join-irreducible, i.e. the vertex $x$
is essential. Indeed, $\cl x$ is join-irreducible iff any subset
$X^{*}\!\subseteq\!\cl x$ such that $\cl x\!=\!\cl{X^{*}}$ contains
a vertex from the equilocality class of $x$, i.e. has at least one
element not contained in $\ucl x$, and this happens iff $\ucl x\!\in\!\latg{\rel}$.
\end{proof}
The importance of essential vertices stems from the following result.
\begin{thm}
\label{thm:isgisom}The associated lattice $\latg{\sigr}$ of the
induced subgraph $\sigr$ corresponding to a collection $\sw$ of
vertices is isomorphic with the associated lattice of $\rel$ iff
$\sw$ contains all essential vertices. 
\end{thm}
\begin{proof}
To prove the only if part, notice that the isomorphism of two lattices
implies that they have the same number of elements. But one has $\latg{\sigr\,}\!=\!\set{X\!\cap\!\sw}{X\!\in\!\isl}$
and $\isl\!\subseteq\!\latg{\rel}$ according to \prettyref{lem:isgr},
consequently $\Bigl|\latg{\sigr\,}\Bigr|\!\leq\!\Bigl|\isl\Bigr|\!\leq\!\Bigl|\latg{\rel}\Bigr|$,
which implies that $\latg{\sigr\,}$ and $\latg{\rel}$ cannot be
isomorphic unless $\isl\!=\!\latg{\rel}$, hence for every $X\!\in\!\latg{\rel}$
there exists $X^{*}\!\subseteq\!\sw$ such that $X\!=\!\du{X^{*}}$.
This should be true in particular for the neighborhood $\du x$ of
any essential vertex $x\!\in\!\ess{\rel}$, and because $\du x$ is
meet-irreducible in that case, $\du x\!=\!\du{X^{*}}$ implies that
$x\!\in\!X^{*}\!\subseteq\!\sw$, proving that indeed $\ess{\rel}\!\subseteq\!\sw$
if $\latg{\sigr\,}$ and $\latg{\rel}$ are isomorphic.

As to the if part, since any lattice element can be written as the
meet of a suitable collection of meet-irreducible elements, and the
latter are, according to \prettyref{eq:clext}, all of the form $\du x$
for some $x\!\in\!\ess{\rel}$, the condition $\ess{\rel}\!\subseteq\!\sw$
means that for any $X\!\in\!\latg{\rel}$ there is a subset $X^{*}\!\subseteq\!\sw$
such that $X\!=\!\du{X^{*}}$, i.e. $\isl\!=\!\latg{\rel}$. This
implies $\cl{X\!\cap\!\sw}\!=\!X$ for $X\!\in\!\latg{\rel}$, according
to \prettyref{lem:isgr}, since $\du X\!\in\!\latg{\rel}\!=\!\isl$.
Consequently, we have a pair of surjective order-preserving maps
\[
\Map{\res}{\latg{\rel}}{\latg{\sigr\,}}X{X\!\cap\!\sw}\textrm{\ensuremath{\qquad}~~and~~\ensuremath{\qquad}\ensuremath{\Map{\boldsymbol{\Updelta}}{\latg{\sigr\,}}{\latg{\rel}}X{\cl X}}}
\]
that are mutually inverse to each other (i.e. lattice isomorphisms),
and which are compatible with the respective duality maps since
\[
\left(\sigr\circ\res\right)\!\left(X\right)\!=\!\sigr\!\left(X\!\cap\!\sw\right)\!=\!\sigr\!\left(X\right)\!=\!\rel\!\left(X\right)\!\cap\!\sw\!=\!\left(\res\circ\rel\right)\!\left(X\right)
\]
as follows once again from \prettyref{lem:isgr}.
\end{proof}
\prettyref{thm:isgisom} highlights the importance of essential vertices
and the corresponding induced subgraph, leading to the following notion.
\begin{defn}
\label{def:raddef}The \emph{radical} $\rad{\rel}$ of the irreducible
graph $\rel$ is the induced subgraph corresponding to the collection
of its essential vertices. 
\end{defn}
\begin{lem}
The radical of an irreducible graph is itself irreducible.
\end{lem}
\begin{proof}
According to \prettyref{def:essdef}, a vertex is essential iff the
smallest element of the associated lattice containing it is join-irreducible.
\end{proof}
Note that the radical need not be connected, and that a $\stg$ (all
of whose vertices are essential) coincides with its radical.
\begin{lem}
\label{lem:radiso}An isomorphism between irreducible graphs induces
an isomorphism between their radicals.
\end{lem}
\begin{proof}
It follows from \prettyref{prop:irrcrit} that essential vertices
are mapped to essential ones by a graph isomorphism, hence the corresponding
induced subgraphs (the radicals) are isomorphic too.
\end{proof}
\begin{thm}
\label{thm:latisom}Two irreducible graphs have isomorphic associated
lattices iff their radicals are isomorphic.
\end{thm}
\begin{proof}
That irreducible graphs having isomorphic radicals have isomorphic
associated lattices follows from \prettyref{thm:isgisom} and \prettyref{lem:graphiso}.
To prove the converse, we shall show that the radical of an irreducible
graph is isomorphic with that of the $\dugr$ of its associated lattice:
this implies the assertion, because an isomorphism between self-dual
lattices induces obviously an isomorphism between their $\dugr$s,
hence between their radicals by \prettyref{lem:radiso}.

Let $L\!=\!\latg{\rel}$ denote the associated lattice of the irreducible
graph $\rel$. The radical $\rad{\dgr}$ of the $\dugr$ $\rel_{L}$
is the unique induced subgraph of $\rel_{L}$ with the least number
of vertices that has an associated lattice isomorphic to $L$. On
the other hand, the associated lattice of $\rad{\rel}$ is also isomorphic
to $L$ according to \prettyref{thm:isgisom}, hence the radical $\rad{\rel}$
is isomorphic to an induced subgraph of $\rel_{L}$ by \prettyref{thm:dugr}
having $\bigl|\ess{\rel}\bigr|$ vertices. Since this equals the number
$\bigl|\ess{\rel_{L}}\bigr|$ of vertices of $\rad{\dgr}$, the graphs
$\rad{\rel}$ and $\rad{\dgr}$ are isomorphic.
\end{proof}
The above results show that there is a one-to-one correspondence between
(isomorphism classes of) stem graphs and self-dual lattices: to each
stem graph corresponds its associated lattice, and to each self-dual
lattice corresponds a stem graph, the radical of its $\dugr$. This
means that in order to classify self-dual lattices one could instead
classify stem graphs, which could prove easier through the use of
suitable graph-theoretic techniques. To this end one should note that,
while a stem graph need not be connected, it is the disjoint union
of its connected components. The connected stem graphs with less
than 4 vertices are displayed on \prettyref{fig:stem}.

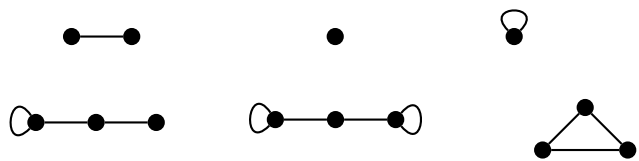
\begin{figure}
\begin{doublespace}
\begin{centering}
~~~~\begin{tikzpicture}[auto,node distance=0.8cm,
every loop/.style={}, thick,
ess/.style={circle,draw, fill, inner sep=0pt,minimum size=2mm},
main/.style={circle,draw, inner sep=0pt,minimum size=5mm}]
\node[ess] (1) {};   
\node[ess] (2) [right of=1] {};
\path[every node/.style={font=\sffamily\small}]        
(1) edge node {} (2);
\end{tikzpicture}~~~~~~~~~~~~~~~\begin{tikzpicture}[auto,node distance=0.8cm,
every loop/.style={}, thick,
ess/.style={circle,draw, fill, inner sep=0pt,minimum size=2mm},
main/.style={circle,draw, inner sep=0pt,minimum size=5mm}]
\node[ess] (1) {};
\end{tikzpicture}~~~~~~~~~~~\begin{tikzpicture}[auto,node distance=0.8cm,
every loop/.style={}, thick,
ess/.style={circle,draw, fill, inner sep=0pt,minimum size=2mm},
main/.style={circle,draw, inner sep=0pt,minimum size=5mm}]
\node[ess] (1) {};
\path[every node/.style={font=\sffamily\small}]        
(1) edge  [loop] node {} (1);
\end{tikzpicture}
\par\end{centering}
\end{doublespace}
\begin{centering}
~~~~~~\begin{tikzpicture}[auto,node distance=0.8cm,
every loop/.style={}, thick,
ess/.style={circle,draw, fill, inner sep=0pt,minimum size=2mm},
main/.style={circle,draw, inner sep=0pt,minimum size=5mm}]
\node[ess] (1) {};   
\node[ess] (2) [right of=1] {};
\node[ess] (3) [right of=2] {};
\path[every node/.style={font=\sffamily\small}]        
(1) edge [out=125, in=225, loop] node  {} (1)
edge node {} (2)
(2) edge node {} (3);
\end{tikzpicture}~~~~~~\begin{tikzpicture}[auto,node distance=0.8cm,
every loop/.style={}, thick,
ess/.style={circle,draw, fill, inner sep=0pt,minimum size=2mm},
main/.style={circle,draw, inner sep=0pt,minimum size=5mm}]
\node[ess] (1) {};   
\node[ess] (2) [right of=1] {};
\node[ess] (3) [right of=2] {};
\path[every node/.style={font=\sffamily\small}]        
(1) edge [out=125, in=225, loop] node  {} (1)
edge node {} (2)
(2) edge node {} (3)
(3) edge [out=50, in=310, loop] node {} (3);
\end{tikzpicture}~~~~~~~\begin{tikzpicture}[auto,node distance=0.8cm,
every loop/.style={}, thick,
ess/.style={circle,draw, fill, inner sep=0pt,minimum size=2mm},
main/.style={circle,draw, inner sep=0pt,minimum size=5mm}]
\node[ess] (1) {};
\node [ess] (0) [below left of=1] {};   
\node[ess] (2) [below right of=1] {};
\path[every node/.style={font=\sffamily\small}]        
(0) edge node {} (1)
edge  node {} (2)
(1) edge node {} (2);
\end{tikzpicture}\caption{\label{fig:stem}All connected stem graphs (up to isomorphism) with
less than 4 vertices.}
\par\end{centering}
\end{figure}

The number of connected stem graphs grows steadily with the number
of vertices, as illustrated in \prettyref{tab:graph_numbers}: according
to the above, this is equal to the number of isomorphism classes of
self-dual lattices with a given number of join-irreducible element.
The table also gives the number of non-isomorphic irreducible graphs
with a modular associated lattice; as we shall see later, these are
those that can (in principle) arise as locality diagrams of conformal
models.

\begin{table}

\begin{centering}
\begin{tabular}{|c||r|r|r|}
\hline 
\# vertices & \# irreducible & \# stem & \# modular \tabularnewline
\hline 
\hline 
1 & 2 & 2 & 2\tabularnewline
\hline 
2 & 2 & 1 & 2\tabularnewline
\hline 
3 & 6 & 3 & 6\tabularnewline
\hline 
4 & 31 & 18 & 24\tabularnewline
\hline 
5 & 230 & 140 & 95\tabularnewline
\hline 
6 & 2683 & 1716 & 439\tabularnewline
\hline 
7 & 50922 & 33448 & 2362\tabularnewline
\hline 
\end{tabular}\caption{\label{tab:graph_numbers}Number of non-isomorphic connected graphs
of different types as a function of the number of vertices (the last
column giving the number of irreducible graphs with a modular associated
lattice).}
\par\end{centering}
\end{table}

The following result provides a partial converse to \prettyref{lem:dugrsub}.
\begin{lem}
\label{lem:dusubirr}If the collection $\sw\!\subseteq\!L$ contains
all join-irreducible elements of the self-dual lattice $L$, then
the induced subgraph $\dgr^{\sw}$ of the duality graph is irreducible.
\end{lem}
\begin{proof}
Since $\sw\!\subseteq\!L$ contains all join-irreducible elements
of $L$, for every $a\!\in\!L$ there exists a subset $A\!\subseteq\!\sw$
such that $\bigvee A\!=\!\duc a$, and because $\dgr^{\sw}\!\left(a\right)\!=\!\set{b\!\in\!\sw}{b\!\leq\!\duc a}$,
one has $A\!\subseteq\!\dgr^{\sw}\!\left(a\right)$. But this implies
that $\duc a\!=\!\bigvee A\!\leq\!\bigvee\dgr^{\sw}\!\left(a\right)\!\leq\!\duc a$,
i.e. $\bigvee\dgr^{\sw}\!\left(a\right)\!=\!\duc a$, and this proves
that indeed $\dgr^{\sw}\!\left(a\right)\!=\!\dgr^{\sw}\!\left(b\right)$
iff $a\!=\!b$.
\end{proof}
While the induced subgraphs of the $\dugr$ corresponding to different
subsets $\ess{\dgr}\!\subseteq\!\sw\!\subseteq\!L$ comprise a full
list of those irreducible graphs whose associated lattice is isomorphic
with $L$, there might exist non-trivial isomorphisms between such
graphs, corresponding to suitable automorphisms of the self-dual lattice
$L$. This issue is settled by the following result.
\begin{prop}
\label{prop:graphiso} Given two subsets $\sw_{1},\sw_{2}\!\subseteq\!L$
of a self-dual lattice $L$ that contain all join-irreducible elements,
the induced subgraphs of the $\dugr$ $\dgr$ corresponding to $\sw_{1}$
and $\sw_{2}$ are isomorphic iff there exists an automorphism $\sigma\!\in\!\aut L$
such that $\sw_{2}\!=\!\sigma\!\left(\sw_{1}\right)$.
\end{prop}
\begin{proof}
If $\sigma\!\in\!\aut L$ satisfies $\sw_{2}\!=\!\sigma\!\left(\sw_{1}\right)$,
then for $a\!\in\!\sw_{1}$ one has
\begin{align*}
\sigma\!\left(\sgr 1_{L}\!\left(a\right)\right)\! & =\!\set{\sigma\!\left(b\right)\!\!}{b\!\in\!\sw_{1},b\!\leq\!\duc a\!}\!=\!\set{b\!\in\!\sw_{2}\!}{b\!\leq\!\sigma\!\left(\duc a\right)}\\
 & =\set{b\!\in\!\sw_{2}\!}{b\!\leq\!\duc{\sigma\!\left(a\right)}\!}\!=\!\sgr 2_{L}\!\left(\sigma\!\left(a\right)\right)
\end{align*}
hence the restriction of $\sigma$ to $\sw_{1}$ provides a graph
isomorphism between the induced subgraphs $\sgr 1_{L}$ and $\sgr 2_{L}$. 

The other way round, if $\sgr 1_{L}$ and $\sgr 2_{L}$ are isomorphic,
i.e. there exists a bijective map $\map{\phi}{\sw_{1}}{\sw_{2}}$
such that $\phi\!\left(\sgr 1_{L}\!\left(a\right)\right)\!=\sgr 2_{L}\!\left(\phi\!\left(a\right)\right)$
for all elements $a\!\in\!\sw_{1}$, then this map can be extended
to an isomorphism $\map{\hat{\phi}}{\latg{\sgr 1_{L}}}{\latg{\sgr 2_{L}}}$
of lattices compatible with the duality maps according to \prettyref{lem:graphiso}.
On the other hand, by \prettyref{thm:isgisom} there exists lattice
isomorphisms $\map{\Upomega_{1}}L{\latg{\sgr 1_{L}}}$ and $\map{\Upomega_{2}}L{\latg{\sgr 2_{L}}}$.
It follows that the composite map $\sigma\!=\!\Upomega_{2}^{-1}\circ\hat{\phi}\circ\Upomega_{1}$
is an automorphism of $L$ that commutes with the duality map, $\sigma\!\left(\duc a\right)\!=\!\duc{\sigma\!\left(a\right)}$
for $a\!\in\!L$, and satisfies $\sw_{2}\!=\!\sigma\!\left(\sw_{1}\right)$
because $\sigma\!\left(a\right)\!=\!\phi\!\left(a\right)$ for $a\!\in\!\sw_{1}$. 
\end{proof}
\global\long\def\latorb#1{\uplambda\left(#1\right)}%
\global\long\def\essorb#1{\upepsilon\left(#1\right)}%

\begin{rem}
Based on the above, one can show that the number of non-isomorphic
irreducible graphs whose associated lattice is isomorphic with a given
the self-dual lattice $L$ equals
\[
\frac{1}{\norm{\aut L}}\sum_{\sigma\in\aut L}2^{\latorb{\sigma}-\essorb{\sigma}}
\]
where $\latorb{\sigma}$, resp. $\essorb{\sigma}$ denotes the number
of orbits of $\sigma\!\in\!\aut L$ on the set of all (resp. join-irreducible)
lattice elements, and the following 'mass formula' holds
\[
\sum_{\rel}\frac{1}{\norm{\aut{\rel}}}=\frac{1}{\norm{\aut L}}\left(\begin{array}{c}
\left|L\right|-\left|\mathcal{J}\right|\\
n-\left|\mathcal{J}\right|
\end{array}\right)
\]
where $\mathcal{J}$ denotes the set of join-irreducible elements
of $L$, and the sum runs over all irreducible graphs with $\left|\rel\right|\!=\!n$
vertices and associated lattice isomorphic to $L$.
\end{rem}
Finally, let's note the following interesting result, which proves
important in physics applications.

\global\long\def\dimf#1{\delta\!\left(#1\right)}%

\begin{lem}
\label{lem:dimbound}The length of a maximal chain ending at $X\!\in\!\latg{\rel}$
cannot exceed the number of essential classes contained in it.
\end{lem}
\begin{proof}
Let $\dimf X$ denote the number of equilocality classes contained
in $X\!\in\!\latg{\rel}$, which makes sense thanks to \prettyref{lem:pclunion}.
We claim that the length of a maximal chain ending at $X$ cannot
exceed $\dimf X$. Indeed, $\dimf Y\!<\!\dimf X$ for any $Y\!\in\!\latg{\rel}$
properly contained in $X$, hence for a maximal chain $Y_{0}\!\subset\!Y_{1}\!\subset\!\cdots\!\subset\!Y_{n}\!=\!X$
ending at $X$ one has $\dimf{Y_{0}}\!<\!\dimf{Y_{1}}\!<\cdots<\dimf X$,
and by taking into account that the values $\dimf X$ are non-negative
integers this proves $n\!\leq\!\dimf X$. Since the associated lattices
of $\rel$ and $\rad{\rel}$ are isomorphic by \prettyref{thm:isgisom},
and all equilocality classes of the latter are essential, applying
this result to the radical $\rad{\rel}$ instead of $\rel$ proves
the claim.
\end{proof}

\global\long\def\ve{\mathtt{V}}%
\global\long\def\FA#1{\vert#1\vert}%

\global\long\def\FB#1{\mathtt{Z}^{2}\!(#1)}%

\global\long\def\FC#1{#1^{{\scriptscriptstyle \flat}}}%

\global\long\def\FD#1{#1^{{\scriptscriptstyle \times}}}%

\global\long\def\FE#1{\mathtt{x}_{#1}}%

\global\long\def\FF#1#2{\mathrm{Fix}_{#1}\left(#2\right)}%

\global\long\def\FI#1{#1_{{\scriptscriptstyle \pm}}}%

\global\long\def\cl#1{\mathscr{C}\negthinspace\ell\!\left(#1\right)}%

\global\long\def\bl#1{\mathcal{B}\ell\!\left(#1\right)}%
\global\long\def\clmap#1{\mathscr{C}\negthinspace\ell{}^{\,(#1)}}%

\global\long\def\bbl#1#2{\mathcal{B}\ell_{#2}\left(#1\right)}%

\global\long\def\ver#1{\mathtt{Ver}_{#1}}%

\global\long\def\stab#1#2{#1_{#2}}%

\global\long\def\cft{\mathscr{C}}%

\global\long\def\gal#1{\upsigma_{#1}}%

\global\long\def\galpi#1{\boldsymbol{\uppi}{}_{#1}}%

\global\long\def\ann#1{\ker\mathfrak{#1}}%

\global\long\def\cent#1#2{\mathtt{C}_{#1}\!\left(#2\right)}%

\global\long\def\ch#1{\boldsymbol{\uprho}_{#1}}%

\global\long\def\qd#1{\mathtt{d}_{#1}}%

\global\long\def\cw#1{\mathtt{h}_{#1}}%

\global\long\def\tcl{\textrm{trivial class}}%

\global\long\def\ccl{\mathtt{C}}%

\global\long\def\zcl{\mathtt{z}}%

\global\long\def\om#1{\omega\!\left(#1\right)}%

\global\long\def\cs#1{\left\llbracket #1\right\rrbracket }%

\global\long\def\coc#1{\vartheta_{\mathfrak{#1}}}%

\global\long\def\sp#1#2{\bigl\langle#1,#2\bigr\rangle}%

\global\long\def\irr#1{\mathtt{Irr}\!\left(#1\right)}%

\global\long\def\sqf#1{#1^{\circ}}%

\global\long\def\fc{\textrm{FC set}}%

\global\long\def\to#1{\boldsymbol{\upsigma}\!\left(#1\right)}%

\global\long\def\v{\mathtt{{\scriptstyle 0}}}%

\global\long\def\wm#1{\mathcal{P}\!\left(#1\right)}%

\global\long\def\fm{\mathtt{N}}%

\global\long\def\du#1{#1^{{\scriptscriptstyle \perp}}}%

\global\long\def\lat{\mathcal{\mathscr{L}}}%

\global\long\def\vera{\mathcal{V}}%

\global\long\def\svera#1{\vera_{#1}}%

\global\long\def\zent#1{\mathtt{Z}\!\left(#1\right)}%

\global\long\def\zenta#1{\mathtt{Z}^{*}\!(#1)}%

\global\long\def\cech#1#2{\boldsymbol{\varpi}_{#1}\!\left(#2\right)}%

\global\long\def\lgr{\boldsymbol{\Lambda}}%

\global\long\def\qgr{\lgr^{\!\intercal}}%

\global\long\def\twa#1{\lgr\!\left(#1\right)}%

\global\long\def\cov#1{#1^{\intercal}}%

\global\long\def\zwa#1{\qgr\!\!\left(\mathfrak{#1}\right)}%

\global\long\def\idch{\boldsymbol{\mathfrak{1}}}%

\global\long\def\usub#1{\mathbf{\boldsymbol{\cup}}#1}%

\global\long\def\zquot#1#2{\mathfrak{#1}/\!#2}%

\global\long\def\extclass#1{#1\ccl}%

\global\long\def\dg#1{\hat{#1}}%

\global\long\def\loc{\lat_{\mathtt{loc}}}%
\global\long\def\intlat{\lat_{\mathtt{int}}}%

\global\long\def\rfc#1{\surd#1}%
\global\long\def\vfc{\mathfrak{o}}%

\global\long\def\sc{\mathcal{\cov{\vfc}}}%

\newcommandx\pwm[2][usedefault, addprefix=\global, 1=]{\boldsymbol{\Uppi}_{#1}^{(#2)}}%

\global\long\def\res#1#2{\mathtt{res}_{#2}^{#1}}%

\global\long\def\clos#1{\langle#1\rangle}%

\global\long\def\mzq#1{\boldsymbol{\mathfrak{Z}}#1}%

\global\long\def\rpc#1#2{#1\!\rightsquigarrow\!#2}%

\global\long\def\um#1{\mathcal{U}\!\left(#1\right)}%
\global\long\def\nsc#1#2#3{\boldsymbol{\mu}_{#3}\!\left(#1,#2\right)}%

\global\long\def\fmpow#1{\mathcal{M}\!\!\left(#1\right)}%

\global\long\def\spp#1#2{\mathcal{P}_{#1}\!\left(#2\right)}%

\global\long\def\qrel{\pitchfork}%

\global\long\def\grl{\textrm{group-like}}%
\global\long\def\fcs{\textrm{FC set}}%
\global\long\def\pfc{\textrm{principal }\fcs}%
\global\long\def\pcl{\textrm{equilocality class}}%

\global\long\def\qlgr{\textrm{quasi-locality graph}}%

\global\long\def\cmp#1{\textrm{(cf. [1], #1)}}%

\global\long\def\ex{\boldsymbol{\text{\upmu}}}%

\section{Locality graphs and FC sets}

Consider a (unitary) rational CFT \cite{BPZ,DiFrancesco-Mathieu-Senechal}.
We shall denote by $\qd p$ and $\mathtt{h}_{p}$ the quantum dimension
and conformal weight of a primary $p$, and by $\fm\!\left(p\right)$
the associated fusion matrix, whose matrix elements are given by the
fusion rules
\begin{equation}
\left[\fm\!\left(p\right)\right]_{q}^{r}=N_{pq}^{r}\label{eq:fmdef}
\end{equation}
$\v$ will denote the vacuum primary for which $\cw{\v}\!=\!0$ and
$\fm\!\left(\v\right)$ is the identity matrix (hence $\qd{\v}\!=\!1$).
Note that, since
\begin{equation}
\fm\!\left(p\right)\fm\!\left(q\right)=\sum_{r}N_{pq}^{r}\fm\!\left(r\right)\label{ver1}
\end{equation}
the fusion matrices generate a commutative matrix algebra over $\mathbb{C}$,
whose irreducible representations, all of dimension $1$, are in one-to-one
correspondence with the primaries. According to Verlinde's famous
result \cite{Verlinde1988}, to each primary $p$ corresponds an irreducible
representation $\ch p$ that assigns to the fusion matrix $\fm\!\left(q\right)$
the complex number
\begin{equation}
\ch p\!\left(q\right)\!=\!\sum_{r}N_{pq}^{r}\frac{\qd r}{\qd p}\mathtt{e}^{2\pi\mathtt{i}\left(\cw p+\cw q-\cw r\right)}\label{eq:chdef}
\end{equation}
Note that $\qd p\!=\!\ch{\v}\!\left(p\right)$, and one has the inequality
\begin{equation}
\FA{\ch p\!\left(q\right)}\!\leq\!\qd q\label{eq:ineq}
\end{equation}

The theory of orbifold deconstruction \cite{Bantay2019a,Bantay2020}
points to the importance of so-called fusion closed sets, or '$\fcs$s'
for short, i.e. sets of primaries containing the vacuum $\v$ and
closed under the fusion product. The collection $\lat$ of all $\fcs$s
(ordered by inclusion) forms a modular lattice \cite{Bantay2020a,Bantay2021}
that is also self-dual, i.e. comes equipped with a duality map sending
each $\fcs$ $\mathfrak{g}\!\in\!\lat$ to $\du{\mathfrak{g}}\!=\!\set p{\ch q\!\left(p\right)\!=\!\qd p\textrm{ for all }q\!\in\!\mathfrak{g}}$,
its so-called 'trivial class', which is itself an $\fcs$. An important
task is to describe the structure of this lattice $\lat$ characterizing
the different possible deconstructions of the model under study.

The problem is that a brute force approach to determine all $\fcs$s
of a given model can be prohibitively difficult. Indeed, the cost
of such a procedure is exponential in the number of primaries, and
it breaks down already for a couple of dozens of primary fields, while
truly interesting examples come with several hundreds, if not thousands
of them. This is where the previous graph-theoretic ideas come to
the rescue, as we shall now explain.

Recall that two primaries of a conformal model are mutually local
if their operator product expansion coefficients are single-valued
functions of separation \cite{Bantay2021a}. In other words, $p$
and $q$ are local if $\cw p\!+\!\cw q$ differs by an integer from
$\cw r$ for each primary $r$ such that $N_{pq}^{r}\!>\!0$; using
Verlinde's formula, this is equivalent to the requirement $\ch q\!\left(p\right)\!=\!\qd p$.
Clearly, mutual locality is a symmetric relation, to which is associated
an undirected graph $\lgr$, the so-called locality graph of the conformal
model, whose vertices correspond to the primary fields, with two vertices
adjacent if the corresponding primaries are mutually local\footnote{There is a closely related notion, that of the 'quasi-locality graph',
whose vertices still correspond to the primary fields, with two of
them adjacent if they saturate the bound \prettyref{eq:ineq}, i.e.
$\FA{\ch p\!\left(q\right)}\!=\!\qd q$; at the level of OPE, this
means that the expansion coefficients are no more necessarily single-valued,
but some power of them is. One can show that the associated lattice
of this graph is a sublattice of $\latg{\lgr}$, with elements corresponding
to maximal Abelian extensions of $\fcs$s (c.f. Section 5 of \cite{Bantay2020a}).}. 

In what follows, we shall make free use of results from \cite{Bantay2020a}
on $\fcs$s and their lattice $\lat$, especially the properties of
the duality map. The basic observation, motivating the present work,
is the following one.
\begin{lem}
\label{lem:cycfc}The vertex-neighborhood
\begin{equation}
\twa{\alpha}\!=\!\set p{\ch{\alpha}\!\left(p\right)\!=\!\qd p}\label{eq:twadef}
\end{equation}
of a primary $\alpha$ in the locality graph is an $\fcs$.
\end{lem}
\begin{proof}
For $p,q\!\in\!\twa{\alpha}$ one has
\[
\sum_{r}N_{pq}^{r}\!\left(\qd r\!-\!\ch{\alpha}\!\left(r\right)\right)\!=\!\sum_{r}N_{pq}^{r}\qd r\!-\!\sum_{r}N_{pq}^{r}\ch{\alpha}\!\left(r\right)\!=\!\qd p\qd q\!-\!\ch{\alpha}\!\left(p\right)\!\ch{\alpha}\!\left(q\right)\!=\!0
\]
hence, after taking real parts,
\[
\sum_{r}N_{pq}^{r}\!\left(\qd r\!-\!\mathsf{Re}\left(\ch{\alpha}\!\left(r\right)\right)\right)=0
\]
All terms on the lhs. are non-negative since  $\mathsf{Re}\!\left(\ch{\alpha}\!\left(r\right)\right)\!\leq\!\FA{\ch{\alpha}\!\left(r\right)\!}\!\leq\!\qd r$,
and this implies that $N_{pq}^{r}\!=\!0$ unless $\ch{\alpha}\!\left(r\right)\!=\!\qd r$,
i.e. $r\!\in\!\twa{\alpha}$. 
\end{proof}
\begin{lem}
\label{lem:cycmin}$\du{\twa{\alpha}}$ is the smallest $\fc$ containing
the primary $\alpha$.
\end{lem}
\begin{proof}
$p\!\in\!\du{\twa{\alpha}}$ precisely when $\ch p\!\left(q\right)\!=\!\qd q$
for all $q\!\in\!\twa{\alpha}$, and since this holds for $\alpha$
by definition, one has $\alpha\!\in\!\du{\twa{\alpha}}$. Moreover,
if $\mathfrak{g}\!\in\!\lat$ contains $\alpha$, then $p\!\in\!\du{\mathfrak{g}}$
implies $\ch p\!\left(\alpha\right)\!=\!\qd{\alpha}$, hence $\ch{\alpha}\!\left(p\right)\!=\!\qd p$,
that is $p\!\in\!\twa{\alpha}$; consequently, $\du{\mathfrak{g}}\!\subseteq\!\twa{\alpha}$,
which gives $\du{\twa{\alpha}}\!\subseteq\!\mathfrak{g}$ by duality,
proving the claim.
\end{proof}
\begin{lem}
\label{lem:cycsub}
\[
\du{\twa{\alpha}}\!=\!\set{\beta}{\twa{\alpha}\!\subseteq\!\twa{\beta}}
\]
\end{lem}
\begin{proof}
$\beta\!\in\!\du{\twa{\alpha}}$ implies $\du{\twa{\beta}}\!\subseteq\!\du{\twa{\alpha}}$
by \prettyref{lem:cycmin}, hence $\twa{\alpha}\!\subseteq\!\twa{\beta}$;
conversely, $\twa{\alpha}\!\subseteq\!\twa{\beta}$ implies $\du{\twa{\beta}}\!\subseteq\!\du{\twa{\alpha}}$,
from which follows $\beta\!\in\!\du{\twa{\alpha}}$ because of $\beta\!\in\!\du{\twa{\beta}}$.
\end{proof}
\begin{cor}
\noindent If $\mathfrak{g}\!\in\!\lat$ is an $\fc$, then
\begin{equation}
\mathfrak{g}=\bigvee_{\alpha\in\mathfrak{g}}\du{\twa{\alpha}}\label{eq:cycdecomp}
\end{equation}
and
\begin{equation}
\du{\mathfrak{g}}=\bigcap_{\alpha\in\mathfrak{g}}\twa{\alpha}\label{eq:cycdec2}
\end{equation}
\end{cor}
\begin{proof}
\prettyref{eq:cycdecomp} is a direct consequence of \prettyref{lem:cycmin},
and \prettyref{eq:cycdec2} follows from it by duality.
\end{proof}

Since, according to \prettyref{lem:dualclosed}, any element of an
associated lattice can be obtained as a meet of vertex neighborhoods,
it follows that $\latg{\lgr}$ is a sublattice of $\lat$. Actually,
a much stronger statement is true.
\begin{thm}
\label{thm:declat}The self-dual lattices $\lat$ and $\latg{\lgr}$
coincide.
\end{thm}
\begin{proof}
By the above, all elements of $\latg{\lgr}$ belong to $\lat$. Since
both lattices are ordered by inclusion, what we have to show is that,
conversely, every $\mathfrak{g}\!\in\!\lat$ belongs to $\latg{\lgr}$,
and that the respective duality maps coincide, i.e. $\du{\mathfrak{g}}\!=\!\twa{\mathfrak{g}}$
for $\mathfrak{g}\!\in\!\lat$. But $\du{\mathfrak{g}}\!=\!\bigcap_{\alpha\in\mathfrak{g}}\twa{\alpha}\!=\!\twa{\mathfrak{g}}$
by \prettyref{eq:cycdec2}, and because $\du{\mathfrak{g}}\!\in\!\fcl$,
one concludes that $\mathfrak{g}\!=\!\du{\left(\du{\mathfrak{g}}\right)}\!=\!\twa{\du{\mathfrak{g}}}$
belongs to $\latg{\lgr}$, proving the claim. 
\end{proof}
\prettyref{thm:declat} was the main motivation behind the theory
presented before, since it makes available the results about associated
lattices of undirected graphs in the study of the $\decl$ $\fcl$.
Not only does it underline the importance of the locality graph, but
also draws attention onto the so-called 'locality diagram' of the
model (the deflation of its locality graph) and the equilocality classes
of primaries. As an illustration, let's note the following interesting
result.
\begin{lem}
\label{lem:eclint}If the quantum dimension of a primary is a rational
integer, then the same is true for all primaries in its equilocality
class.
\end{lem}
\begin{proof}
According to Corollary 11 of \cite{Bantay2020a}, the set $\ifc\!=\!\set p{\qd p\!\in\!\mathbb{Z}}$
of all primaries with integer quantum dimension is an $\fcs$, hence
it belongs to $\latg{\lgr}$ according to \prettyref{thm:declat}.
Consequently, $\ifc$ is a union of equilocality classes by \prettyref{cor:ecldecomp},
hence $p\!\in\!\ifc$ implies that all primaries in the equilocality
class of $p$ also belongs to $\ifc$.
\end{proof}
\begin{rem}
Let us mention that, using some elementary Galois theory together
with the results of \cite{Bantay2020a}, one can prove the following
generalization of \prettyref{lem:eclint}: the algebraic number field
generated by the quantum dimension of a primary is the same for all
primaries in the same equilocality class.
\end{rem}
As to the locality diagram, i.e. the deflation of the locality graph,
it is clear that its knowledge, together with that of the primary
content of the individual equilocality classes (i.e. the deflation
map) does completely determine the locality graph itself. But the
associated lattice of a graph and of its deflation are isomorphic
according to \prettyref{thm:defliso}, and this means that the locality
diagram does determine the structure (as a self-dual lattice) of the
$\decl$ $\fcl$. In particular, fairly different conformal models
could have isomorphic $\decl$s provided their locality diagrams are
the same. From a practical point, since the locality diagram has usually
much less vertices than the locality graph, these observations lead
to a dramatic simplification in the computation of the $\decl$.

But there is more to all this, as computational evidence suggests
that many different conformal models of similar origin have the same
(up to isomorphism) locality diagram: for example, all unitary Virasoro
minimal models, except for two, share the same 'generic Virasoro diagram',
while for superconformal minimal models one has two such generic diagrams.
The situation gets more complicated for other classes (like parafermionic
or Gepner models), but there is still a clear pattern in the structure
of the diagrams, and this suggests that locality diagrams can provide
a (coarse) classification of conformal models\footnote{Actually, since the notion of mutual locality does rely exclusively
on the modular data of the conformal model, locality diagrams make
sense in the larger context of Modular Tensor Categories \cite{Turaev,Bakalov-Kirillov},
and can be used to classify the latter.}.

Finally, let us note that locality diagrams of conformal models have
a fairly restricted structure, since each 
\begin{enumerate}
\item is irreducible , being the deflation of the locality graph; 
\item has an associated lattice that has to be modular  according to Theorem
2 of \cite{Bantay2020a};
\item has a 'universal' vertex adjacent to every vertex (including itself)
that corresponds to the equilocality class containing the vacuum.
\end{enumerate}
Of all the irreducible graphs of a given size, these criteria select
out a small number of graphs: for example, the number of non-isomorphic
irreducible graphs of a given size satisfying suitable combinations
of these criteria can be read off \prettyref{tab:graphstat}. What
is really interesting is that computational evidence suggests that
only a small fraction of all those graphs that satisfy these criteria
does actually correspond to the locality diagram of some conformal
model. Actually, an even stronger statement seems to be true: there
is a one-to-one correspondence between locality diagrams of conformal
models and their $\decl$s. This observation is indeed intriguing,
as there are in principle several different irreducible graphs with
a given associated lattice satisfying all the above criteria, but
only one of these seems to be realized as the locality diagram of
some suitable conformal model. Clarifying this issue could lead to
a better understanding of conformal models (and Modular Tensor Categories).

\begin{table}
\centering{}%
\begin{tabular}{|c||r|r|r|c|}
\hline 
\# vertices & 1 & 1 \& 2 & 1 \& 2 \& 3 & known locality diagrams \tabularnewline
\hline 
\hline 
3 & 6 & 6 & 3 & 1\tabularnewline
\hline 
4 & 31 & 24 & 8 & 4\tabularnewline
\hline 
5 & 230 & 95 & 27 & 2\tabularnewline
\hline 
6 & 2683 & 439 & 98 & 3\tabularnewline
\hline 
7 & 50922 & 2362 & 443 & 2\tabularnewline
\hline 
\end{tabular}\caption{\label{tab:graphstat}Statistics of connected graphs with a given
number of vertices satisfying different combinations of the criteria
for locality diagrams, namely criterion 1 (irreducibility), criterion
2 (modularity of the associated lattice) and criterion 3 (existence
of universal vertex), the last column giving the number of known locality
diagrams of the given size.}
\end{table}

Actually, there is a fourth, more subtle criterion for the realizability
of an irreducible graph $\rel$ as the locality diagram of some conformal
model, that follows from Lemma 9 of \cite{Bantay2020a}: there should
exits a positive function $\boldsymbol{\mu}$ on the set of vertices
(in case of locality diagrams, $\boldsymbol{\mu}$ equals the sum
of dimensions squared of the primaries in the corresponding equilocality
class) such that the product
\[
\left(\sum_{x\in X}\boldsymbol{\mu}\!\left(x\right)\right)\left(\sum_{y\in\rel(X)}\boldsymbol{\mu}\!\left(y\right)\right)
\]
is the same for every $X\!\in\!\latg{\rel}$. This is equivalent to
the existence of a strictly positive solution for a set of quadratic
equations determined by the graph, and that this criterion is meaningful
is exemplified by the graphs depicted on \prettyref{fig:non-pos},
which satisfy all the above criteria except for this last one. This
allows to effectively reduce the number of allowed diagrams with 3
vertices from 3 to 1, and from 8 to 5 in case of 4 vertices. Unfortunately,
this 'positivity' criterion can be prohibitively difficult to check
for graphs with 5 or more vertices.

\begin{figure}
~~~~~~~~~~~~~~\begin{tikzpicture}[auto,node distance=1.5cm,
every loop/.style={}, thick,null/.style={coordinate},
main/.style={circle,draw,fill,inner sep=3 pt}]
\node[main] (1) {};
\node[null] (3) [below  of=1] {};
\node [main] (0) [left of=3] {};   
\node[main] (2) [right of=3] {};

\path[every node/.style={font=\sffamily\small}]        
(0) edge node {} (1)
edge  node {} (2)
(1) edge node {} (2)
(2) edge [out=50, in=310, loop] node {} (2);
\end{tikzpicture}~~~~~~~~~~~~~~~~\begin{tikzpicture}[auto,node distance=1.5cm,
every loop/.style={}, thick,null/.style={coordinate},
ess/.style={circle,draw,fill,inner sep=3 pt}]
\node [ess] (0) {}; 
\node[ess] (1)  [below of=0] {};    
\node[ess] (4) [left of=1] {};
\node[ess] (5) [right of=1] {}; 
\path[every node/.style={font=\sffamily\small}]        
(1) edge node {} (0)
edge node {} (4)
edge node {} (5)
(0) edge [loop] node {} (0)
edge  node  {} (4)
edge  node {} (5)
(5) edge [out=50, in=310, loop] node {} (5); 
\end{tikzpicture}\caption{\label{fig:non-pos}Examples of irreducible graphs satisfying the
first three criteria, but not the fourth (positivity) criterion. }
\end{figure}
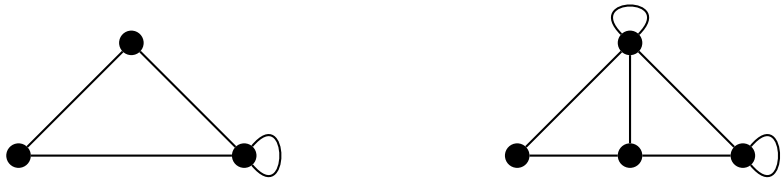

An interesting numerical attribute of locality diagrams is their 'dimension',
the length of the longest chain between the minimal and the maximal
elements of their associated lattices, whose existence is guaranteed
by the modularity of the latter. It follows from \prettyref{lem:dimbound}
that the dimension cannot exceed the number of essential vertices,
hence it can be quite small even for large irreducible graphs. In
particular, there are infinitely many locality diagrams of dimension
2, corresponding to holomorphic orbifolds whose twist group is cyclic
of prime order.

\section{Summary and outlook}

As we have seen, there is an intimate relationship between undirected
graphs (with loops) and self-dual lattices: each graph determines
a lattice, its associated lattice, and conversely, to each self-dual
lattice corresponds a graph, its duality graph, whose associated lattice
is the original one. While several (as a matter of fact, infinitely
many) different graphs share the same associated lattice, one has
a nice description of all these using \prettyref{thm:defliso}, \prettyref{lem:dugrsub}
and \prettyref{thm:isgisom}. The utility of this for physics stems
from the observation that the $\decl$ of a conformal model, of prime
interest for orbifold deconstruction, is nothing but the associated
lattice of the locality graph of the model, making available graph-theoretic
tools in the study of the $\decl$. Not only does one get access to
graph theory as a powerful tool, but this connection brings to front
such new concepts as locality diagrams and equilocality classes in
the description of conformal models (and more generally, of Modular
Tensor Categories).

There are many interesting questions related to this circle of ideas
that merit further elaboration. An obvious one is to understand the
pattern that governs the structure of the locality diagrams of models
from some specified class, e.g. parafermionic or superconformal. Closely
related to the previous question is to explain why there seems to
be a unique locality diagram for all conformal models with the same
deconstruction lattice. Another problem is related to the study of
the possible degenerations of locality diagrams: for example, in case
of unitary Virasoro minimal models, all diagrams but two are isomorphic
with each other, but even the two exceptions may be understood as
degenerations of the generic diagram, due to the fact that the corresponding
models have too few primaries to effectively fill all the equilocality
classes of the latter. Similar phenomena show up in all cases known
to us, making it an interesting point to describe the different possible
degenerations of a given diagram. There is also the question of what
are common the attributes of primaries from the same equilocality
class, and to which extent do locality diagrams offer a meaningful
classification scheme for conformal models. We strongly believe that
the study of these questions could lead to important new developments
in the field.

\bibliographystyle{plain}

\begin{thebibliography}{10}

\bibitem{Bakalov-Kirillov}
B.~Bakalov and A.A. Kirillov.
\newblock {\em Lectures on {T}ensor {C}ategories and {M}odular {F}unctors},
  volume~21 of {\em University Lecture Series}.
\newblock AMS, Providence, 2001.

\bibitem{Balakrishnan1997}
V.~K. Balakrishnan.
\newblock {\em Graph Theory}.
\newblock McGraw-Hill, 1997.

\bibitem{Bantay2021}
P.~Bantay.
\newblock Character rings and fusion algebras.
\newblock 768:179--186.

\bibitem{Bantay2003b}
P.~Bantay.
\newblock The kernel of the modular representation and the galois action in
  rcft.
\newblock {\em Commun. Math. Phys.}, 233:423--438, 2003.

\bibitem{Bantay2003a}
P.~Bantay.
\newblock Symmetric products, permutation orbifolds and discrete torsion.
\newblock {\em Lett. Math. Phys.}, 63:209--218, 2003.

\bibitem{Bantay2019a}
P.~Bantay.
\newblock A short guide to orbifold deconstruction.
\newblock {\em SIGMA 15, 027}, 2019.

\bibitem{Bantay2020a}
P.~Bantay.
\newblock {FC} sets and twisters: the basics of orbifold deconstruction.
\newblock {\em Commun. Math. Phys. 379 (2), 693-721}, 2020.

\bibitem{Bantay2020}
P.~Bantay.
\newblock Orbifold deconstruction: a computational approach.
\newblock {\em Contemp. Math. 753, 1-15}, 2020.

\bibitem{Bantay2021a}
P.~Bantay.
\newblock Deconstruction hierarchies and locality diagrams of conformal models.
\newblock {\em Modern Physics Letters A}, 36(37):2150255, 2021.

\bibitem{BPZ}
A.A. Belavin, A.M. Polyakov, and A.B. Zamolodchikov.
\newblock Infinite conformal symmetry in two-dimensional {Q}uantum {F}ield
  {T}heory.
\newblock {\em Nucl. Phys.}, B241:333--380, 1984.

\bibitem{Bollobas2002}
B.~Bollob\'{a}s.
\newblock {\em Modern Graph Theory}.
\newblock Springer New York, 2002.

\bibitem{Dijkgraaf_disctors}
R.~Dijkgraaf.
\newblock Discrete torsion and symmetric products.
\newblock {\em hep-th/9912101}, 1999.

\bibitem{elliptic_genera}
R.~Dijkgraaf, G.~Moore, E.~Verlinde, and H.~Verlinde.
\newblock {\em Commun. Math. Phys.}, 185:197--209, 1997.

\bibitem{DV3}
R.~Dijkgraaf, C.~Vafa, E.~Verlinde, and H.~Verlinde.
\newblock The operator algebra of orbifold models.
\newblock {\em Commun. Math. Phys.}, 123:485, 1989.

\bibitem{Dixon_orbifoldCFT}
L.J. Dixon, D.~Friedan, E.J. Martinec, and S.H. Shenker.
\newblock The conformal field theory of orbifolds.
\newblock {\em Nucl. Phys.}, B282:13--73, 1987.

\bibitem{Dixon_orbifolds1}
L.J. Dixon, J.A. Harvey, C.~Vafa, and E.~Witten.
\newblock Strings on orbifolds.
\newblock {\em Nucl. Phys.}, B261:678--686, 1985.

\bibitem{Dixon_orbifolds2}
L.J. Dixon, J.A. Harvey, C.~Vafa, and E.~Witten.
\newblock Strings on orbifolds 2.
\newblock {\em Nucl. Phys.}, B274:285--314, 1986.

\bibitem{DiFrancesco-Mathieu-Senechal}
P.~Di Francesco, P.~Mathieu, and D.~S{\' e}n{\' e}chal.
\newblock {\em {C}onformal {F}ield {T}heory}.
\newblock Springer, New York, 1997.

\bibitem{FLM1}
I.~Frenkel, J.~Lepowsky, and A.~Meurman.
\newblock {\em Vertex {O}perator {A}lgebras and the {M}onster}.
\newblock Academic Press, New York, 1988.

\bibitem{Gratzer2011}
G.~Gratzer.
\newblock {\em Lattice Theory}.
\newblock Birkhauser Verlag, Basel, 2011.

\bibitem{Turaev}
V.G. Turaev.
\newblock {\em Quantum Invariants of Knots and 3-Manifolds}, volume~18 of {\em
  Studies in Mathematics}.
\newblock de Gruyter, Berlin, 1994.

\bibitem{Verlinde1988}
E.~Verlinde.
\newblock {\em Nucl. Phys. B}, 300:360, 1988.

\end{thebibliography}

\end{document}